\documentclass[12pt,twoside]{article}
\usepackage{amsmath,amssymb,color,epsfig,epic,subfigure,amsfonts,hhline,graphicx}
\usepackage{ulem, cancel}

\def\N{{\rm I\kern-.25em N}}
\def\R{{\rm I\kern-.25em R}}

\def\Q{{\rm I\kern-.25em Q}}

\def\ol{\overline}

\def\Tr{\text{Tr}}

   \def\L{\Lambda}             
               
\def\p{\pi}

\usepackage{ntheorem,lipsum}

\newtheorem{thm}{Theorem}[section]
\newtheorem{lem}[thm]{Lemma}
\newtheorem{prop}[thm]{Proposition}
 {
}
\newtheorem{cor}[thm]{Corollary}

\theorembodyfont{\upshape}
\newtheorem{define}[thm]{Definition}
\newtheorem{rem}[thm]{Remark}
\newtheorem{ex}[thm]{Example}
\newenvironment{proof}{\bigskip\par\noindent{\it Proof:}}{$\square$\newline\vspace*{0.2cm}}

\makeatletter \@addtoreset{equation}{section}

\makeatother \makeatother
\begin{document}

\title{Quantum Markov chains associated with open quantum random walks }
\author{}
\date{}
\maketitle

\date{}
\maketitle
 \baselineskip=18pt
\vspace{-2cm}

\begin{center}
\parbox{5in}

\noindent{Ameur Dhahri}

\noindent {\small Department of Mathematics, Chungbuk National
University,\\ Chungdae-ro, Seowon-gu, Cheongju,
Chungbuk 28644, Korea\\[-0.1cm]
e-mail: ameur@chungbuk.ac.kr} \vskip 0.5 true cm

\noindent{ Chul Ki Ko}

\noindent {\small University College, Yonsei University,\\ 85 Songdogwahak-ro, Yeonsu-gu, Incheon 21983, Korea\\[-0.1cm]
e-mail: kochulki@yonsei.ac.kr} \vskip 0.5 true cm

 \noindent {Hyun Jae Yoo}\footnote{Corresponding author}

\noindent {\small Department of Applied Mathematics, Hankyong National University,\\
327 Jungang-ro, Anseong-si, Gyeonggi-do 17579, Korea\\[-0.1cm]
e-mail: yoohj@hknu.ac.kr} \vskip 0.2 true cm
\end{center}

\begin{abstract}
{In this paper we construct (nonhomogeneous) quantum Markov chains associated with open quantum random walks. The quantum Markov chain, like the classical Markov chain, is a fundamental tool for the investigation of the basic properties such as reducibility/irreducibility, recurrence/transience, accessibility, ergodicity, etc, of the underlying dynamics. Here we focus on the discussion of the reducibility and irreducibility of open quantum random walks via the corresponding quantum Markov chains. Particularly we show that the concept of reducibility/irreducibility of open quantum random walks in this approach is equivalent to the one previously done by Carbone and Pautrat. We provide with some examples. We will see also that the classical Markov chains can be reconstructed as quantum Markov
chains. }
\end{abstract}

\vspace{0.2cm} {\bf Key words}: Open quantum random walks, quantum Markov chain, transition expectation, reducibility, irreducibility, classical Markov chain.

\vspace{0.2cm} {\bf Mathematical subject classification (2000)}: 60J10, 46L55, 37A30, 82C10, 82C41.

\setlength{\hoffset}{-0.2in}

\baselineskip=18pt
\section{Introduction}
The purpose of this paper is to construct the quantum Markov chains (QMCs hereafter) associated with open quantum random walks (OQRWs) and investigate some interesting properties. Here we focus on the reducibility and irreducibility of QMCs for OQRWs.

The OQRWs were introduced by Attal, {\it et al.}  in \cite{AG-PS, APSS, APS}
 to model the quantum random walks. In particular, the OQRWs were developed
 to formulate the dissipative quantum computing algorithms and dissipative quantum state preparation.
 In that paper the authors introduced the concept of quantum trajectories. This is a repeated process of
 completely positive mapping on a state (an evolution of OQRW, see the next section for the detail) and
  a measurement of the position. By this they constructed a (classical) Markov chain.
  Using this Markov chain, Attal {\it et al.} established a central limit theorem
 for the asymptotic behavior of the OQRWs \cite{AG-PS}.

Recently the dynamical behavior of OQRWs drew many interests and
some works have been done for the ergodicity, hitting times,
recurrence, reducibility, etc, of OQRWs \cite{CP, DM, L, LS}. In
\cite{DM}, Dhahri and Mukhamedov constructed the QMCs for the OQRWs and investigated recurrence and accessibility
of the QMC. On the other hand the QMC was introduced by Accardi
\cite{Ac1, Ac2, Ac3} and further developed \cite{AK1,
AK2}, and has found several applications. See e.g.,  \cite{AF1, AF2,
AMS1, AMS2, AMS3} and references therein. The main ingredient for
the QMC is the transition expectation, which is a completely
positive map and it is a quantum version of the transition matrix for the classical Markov chains \cite{AK1, AK2}. See Section 3 for the
details. Accardi and Koroliuk, after defining the QMC, developed
the quantum versions of reducibility and irreducibility,
accessibility, recurrence and transience \cite{AK1, AK2}. In this
paper we adopt the construction of QMCs for OQRWs done in \cite{DM}
with some modifications. A remarkable point in our construction is that  we have introduced the sub-Markovian transition expectations, contrasting to the fact that it is generally required to have Markovianity for the transition expectations. The Markovianity is recovered when we talk together with the initial conditions and the transition expectations. It seems that this approach is necessary when we try to recover the original dynamics. Another typical notice in our construction is that we have considered the nonhomogeneous quantum Markov
chains instead of homogeneous ones. This is also necessary to recover the original dynamics unless we start with an initial state which is invariant under the dynamics. After constructing the QMCs associated with OQRWs, we study
the reducibility and irreducibility of the OQRWs in the language of the constructed QMCs. We give some sufficient conditions for reducibility/irreducibility providing with some examples. We separately show that the classical Markov chains are reconstructed by the quantum Markov chains and the classical reducibility/irreducibility can be studied by the language of QMCs.  

Let us briefly overview the contents of this paper. In Section 2,
we recall the definition of OQRWs as defined in \cite{APSS}.
Section 3 summarizes the construction of QMCs. Section 4 is the main part of this paper.
We construct the nonhomogeneous QMCs associated with OQRWs using (sub-Markovian) transition expectations. We then develop a characterization for the reducibility/irreducibility (Theorem \ref{thm4.8}) and give some sufficient conditions for reducibility (Theorem \ref{thm4.9}) and irreducibility (Theorem \ref{thm:irreducibility}). Section 5 is devoted to the examples. We construct some examples of reducible and irreducible OQRWs in 1-dimensional integer lattice. We also investigate the relation with classical Markov chains. In Subsection \ref{subsec:CMC} we construct a QMC for a given classical Markov chain. We show that our construction is natural in the sense that it realizes the original classical Markov chain. We then compare the reducibility and irreducibility properties viewed in quantum and classical Markov chains. Finally, in the Appendix we compare with the previous results on the reducibility/irreducibility for OQRWs studied by Carbone and Pautrat \cite{CP}. In fact, it turns out that the concepts of reducibility/irreducibility of OQRWs given in \cite{CP} and in the present paper are equivalent.  

\section{Open quantum random walks}

In this section we briefly introduce the open quantum random walks.

Let $\mathcal{K}$ be a separable Hilbert space with an orthonormal
basis $\{|i\rangle\}_{i\in\Lambda}$ indexed by the vertices of
some graph $\Lambda$. Here the set $\Lambda$ of vertices may be
finite or countably infinite. Let $\mathcal{H}$ be another
separable Hilbert space, which will describe the degrees of
freedom given at each point of $\Lambda$. We consider the space
$\mathcal{H}\otimes\mathcal{K}$.

For each pair $i, j\in\Lambda$ we give a bounded linear operator
$B^i_j$ on $\mathcal{H}$. This operator stands for the effect of
passing from $j$ to $i$. We assume that for each $j$
\begin{equation}
\sum_i{B^i_j}^*B^i_j=I,\label{2.1}
 \end{equation}
 where the series is strongly convergent to the identity operator $I$. This constraint means that the sum of all
  the effects leaving site $j$ is $I$. We dilate the operators $B^i_j$ on $\mathcal{H}$
  as operators on $\mathcal{H}\otimes\mathcal{K}$ by defining
 $$M^i_j=B^i_j\otimes|i\rangle\langle j|.$$
 The operators $M^i_j$ encodes exactly the idea that while passing from $j$ to
 $i$ on the space, the effect is the operator $B^i_j$ on $\mathcal{H}$. By \eqref{2.1}, it is easy to see that
 \begin{equation}
\sum_{i,j}{M^i_j}^*M^i_j=I.\label{2.2}
 \end{equation}
Using the operators $\{M^i_j\}_{i,j}$, define a completely
positive map on ${\mathcal I}_1(\mathcal{H}\otimes\mathcal{K})$,
the ideal of trace class operators, by:
\begin{equation}
\mathcal{M}(\rho)=\sum_i\sum_jM^i_j\rho{M^i_j}^*. \label{2.3}
 \end{equation}
We consider density matrices on $\mathcal{H}\otimes\mathcal{K}$ of the particular form
 $$\rho=\sum_i\rho_i\otimes|i\rangle\langle i|,$$
 where for each $i\in\Lambda$, $\rho_i$ is a positive definite trace class operator
 and satisfies $\sum_i\text{Tr}(\rho_i)=1$. For a given initial state of such form, the OQRW
 is defined by the completely positive map $\mathcal{M}$:
 \begin{equation}
\mathcal{M}(\rho)=\sum_i\left(\sum_jB^i_j\rho_j{B^i_j}^*\right)\otimes|i\rangle\langle i|. \label{2.4}
 \end{equation}
 Hence a measurement of the position in $\mathcal{K}$ would give a probability
 $\sum_j\text{Tr}(B^i_j\rho_j{B^i_j}^*)$  to find out the particle at site $i$.
 The OQRW is a repeated operation of the completely positive map $\mathcal M$.
The two-step evolution, for instance, is of the form
$$\mathcal{M}^2(\rho)=\sum_i\sum_j\sum_kB^i_jB^j_k\rho_k{B^j_k}^*{B^i_j}^*\otimes|i\rangle\langle i|.$$

\section{Quantum  Markov  chains}

In this section we briefly recall the definitions of quantum
Markov chains \cite{AK1, AK2, DM, Lu} and (ir)reducibility
\cite{AK1, AK2}.

\subsection{Quantum Markov chains}

Let $\mathbb{Z}_+$ be the set of all nonnegative integers. Let
$\mathcal{B}$ be a von Neumann subalgebra of $\mathcal{B}(h)$, the space of all bounded linear operators on a separable Hilbert space $h$.
 For any bounded $\Lambda\subset\mathbb{Z}_+$, let
 \begin{equation}\label{eq:finite_tensor_product}
 \mathcal{A}_\Lambda:=\bigotimes_{i\in\Lambda}\mathcal{A}_i,
 ~~\mathcal{A}_i=\mathcal{B},
\end{equation}
be the finite tensor product of von Neumann algebras and
\begin{equation}\label{eq:infinite_tensor_product}
\mathcal{A}:= \bigotimes_{i\in\mathbb{Z}_+}\mathcal{A}_i
 \end{equation}
 be the infinite tensor product of von Neumann algebras \cite{BR, Na}.
  For each $i\in
\mathbb{Z}_+$,
let $J_i$ be the embedding homomorphism
$$J_i: \mathcal{B} \hookrightarrow I_0\otimes I_1\otimes\cdots\otimes I_{i-1}
\otimes\mathcal{B}\otimes
I_{i+1}\otimes\cdots=:I_{i-1]}\otimes\mathcal{B}\otimes I_{[i+1}$$
defined by $$J_i(a)=I_{i-1]}\otimes a\otimes I_{[i+1},\quad\forall
a\in\mathcal{B}.$$ For each $\Lambda\subset\mathbb{Z}_+$, we
identify $\mathcal{A}_\Lambda$ as a subalgebra of $\mathcal{A}$.
We denote $\mathcal{A}_{n]}$ the subalgebra of $\mathcal{A}$, generated by
 the first $(n+1)$ factors, i.e., by the elements of the form
 $$a_{n]}=a_0\otimes a_1\otimes\cdots\otimes a_n\otimes I_{[n+1}=J_0(a_0)J_1(a_1)\cdots J_n(a_n)$$
with $a_0,a_1,\cdots, a_n\in\mathcal{B}$.

A bilinear map $\mathcal{E}$ from $\mathcal{B}\otimes\mathcal{B}$
to $\mathcal{B}$ is called a transition expectation if it is
completely positive and sub-Markovian in the sense that \cite{AW}
\begin{equation}\label{eq:sub-Markovian}
\mathcal{E}(I\otimes
I)\le I.  
\end{equation}
\begin{rem}\label{rem:sub-Markovian}
In the literature, it is required in general the Markovian property, i.e., $\mathcal{E}(I\otimes
I)= I$, to define quantum Markov chains. The sub-Markovian condition \eqref{eq:sub-Markovian} is definitely weaker than the Markovian condition. We emphasize, however, that when we apply the QMCs to special models, like the OQRWs of the present model, it is generally required to use sub-Markovian transition expectations in order to properly recover the original dynamics. Nonetheless, as will be seen in Definition \ref{def:qmc}, since the QMCs are always defined by a pair of initial states and transition expectations, we have a room to recover the Markovian property, and this really works in the present model. We therefore impose the Markovian property only when we speak together with initial states and transition expectations. 
\end{rem}
Given a sequence of transition expectations $(\mathcal
E^{(n)})_{n\ge 0}$, for each $m\ge 0$ we will define a
(unique) completely positive, sub-Markovian map $E_{m]}:
\mathcal{A} \rightarrow \mathcal{A}_{m]}$. Since we have sub-Markovian transition expectations in general, we need some auxiliary preparation.
\begin{lem}\label{lem:cp-sub-Markovian map}
For each $n\ge 0$, there exists a (unique) nonnegative element, denoted by $ \ol{b}(n)\in \mathcal B$, such that $ \ol{b}(n)\le I$ and 
\[
\lim_{k\to \infty}\mathcal E^{(n)}(I\otimes \mathcal E^{(n+1)}(I\otimes \cdots\otimes\mathcal E^{(n+k)}(I\otimes I)))= \ol{b}(n).
\]
In the case that the transition expectations $(\mathcal
E^{(n)})_{n\ge 0}$ are Markovian, $ \ol{b}(n)=I$.
\end{lem}
\begin{proof}
The second statement is trivial. 
Define $a_k^{(n)}:=\mathcal E^{(n)}(I\otimes \mathcal E^{(n+1)}(I\otimes \cdots\otimes\mathcal E^{(n+k)}(I\otimes I)))$. By \eqref{eq:sub-Markovian}, $\{a_k^{(n)}\}_{k\ge 0}$ is a sequence of positive decreasing operators on $\mathcal B$. Hence by Vigier's Theorem \cite{RSz.-N} it strongly converges to a nonnegative element, say $ \ol{b}(n)\in \mathcal B$.  
\end{proof}
In order to define $E_{m]}:
\mathcal{A} \rightarrow \mathcal{A}_{m]}$, first for an element $a_{n]}=a_0\otimes\cdots\otimes a_n\otimes I_{[n+1}\in \mathcal A_{n]}$, $n\ge m$, we define 
\begin{eqnarray}\label{3.2}
&&E_{m]}(a_{n]}):=a_0\otimes\cdots\otimes a_{m-1}\otimes\mathcal{E}^{(m)}
(a_m\otimes\mathcal{E}^{(m+1)}(a_{m+1}\otimes\cdots\nonumber\\
&&\hskip 3true cm\otimes\mathcal{E}^{(n)}(a_n\otimes  \ol{b}(n+1)
 ))).
\end{eqnarray}
And for $a=a_0\otimes a_{1}\otimes\cdots\in \mathcal A$, we let
\begin{equation}\label{eq:cp_conditional_expectation}
E_{m]}(a):=\lim_{n\to \infty} E_{m]}(a_{n]}).
\end{equation}
See \cite{Ac1, Ac2, Ac3, AW}.  

Suppose that a sequence of transition expectations $(\mathcal E^{(n)})_{n\ge 0}$ and a state $\phi_0$ on $\mathcal{B}$ are given. We define a positive definite functional $\phi$ on $\mathcal A$ by 
\begin{equation}\label{eq:qmc}
\phi(a):=\phi_0(E_{0]}(a)),\quad a\in \mathcal A. 
 \end{equation}
Notice that by \eqref{eq:sub-Markovian} and Lemma \ref{lem:cp-sub-Markovian map}, and from the definition of $E_{0]}$ in \eqref{3.2} and \eqref{eq:cp_conditional_expectation}, $\phi$ is sub-Markovian, meaning that $\phi(I\otimes I\otimes\cdots)\le 1$.
\begin{define}\label{def:qmc}
(i) A pair $\left(\phi_0, (\mathcal E^{(n)})_{n\ge 0}\right)$ of a state 
$\phi_0$ on $\mathcal{B}$ and  a sequence of transition expectations $(\mathcal E^{(n)})_{n\ge 0}$ is called a Markov pair if the positive definite functional $\phi$ in \eqref{eq:qmc} defines a state on $\mathcal A$, i.e., it is Markovian in the sense that 
\[
\phi(I\otimes I\otimes\cdots)=1.
\]
(ii) A Markov pair $\left(\phi_0, (\mathcal E^{(n)})_{n\ge 0}\right)$, or alternatively the state $\phi$ in \eqref{eq:qmc} defined by the pair, is called a nonhomogeneous QMC with initial state $\phi_0$.  
When $\mathcal{E}^{(n)}=\mathcal{E}$ for all $n$, we say that the QMC is  homogeneous.
\end{define}
\begin{rem}
The state $\phi$ in the Definition \ref{def:qmc} was called a generalized Markov chain in \cite{AW}. 
\end{rem}
We introduce a typical way of defining the transition expectations
\cite{AF, AW}. Denote by
 $\text{Tr}_i$, $i=1,2$ the partial traces on $\mathcal B\otimes \mathcal B$ defined by
 $$\text{Tr}_1(a\otimes b)=\text{Tr}(a)b,\quad \text{Tr}_2(a\otimes b)=\text{Tr}(b)a.$$
 Let $\{K_i\}_{i\in\mathbb{Z}_+}$
 be a set of Hilbert-Schmidt operators on   $\mathcal{B}\otimes\mathcal{B}$  satisfying
\begin{equation}\label{3.3}
\sum_i\|K_i\|^2<\infty\text{ and }
\sum_i\text{Tr}_2(K_iK_i^*)\le I.
 \end{equation}
Then a transition expectation is defined by \cite{AF, AW}
  \begin{equation}\label{3.4}
       \mathcal{E}(a):=\sum_i \text{Tr}_2(K_iaK_i^*),\quad a\in\mathcal{B}\otimes\mathcal{B}.
 \end{equation}
In this paper, the transition expectations of the type in
\eqref{3.4} with suitably chosen operators $\{K_i\}$ will play a
central role. We notice that in the literature, the equality was required in the equation \eqref{3.3} to define transition expectations satisfying the equality in \eqref{eq:sub-Markovian}. By relaxing it to an inequality as above, it will define a transition expectation which is sub-Markovian in the sense of \eqref{eq:sub-Markovian}. In the applications, like in the present model, the sub-Markovian property is natural. We remark also that Park and Shin computed the
dynamical entropy of generalized QMC constructed by transition
expectations of the type in \eqref{3.4} \cite{P, PS}.

\subsection{Reducible and irreducible QMCs}

In this subsection, we discuss the reducibility and irreducibility of  QMCs.

 We introduce the notion of the reducibility of QMC \cite{AK1, AK2}. Given a projection
 $p\in \mathcal B$ and any $n\in \mathbb Z_+$, we denote
\begin{equation}\label{eq:projection_prototype}
p_{[n}:=I\otimes I \otimes\cdots I \otimes \overset{n\text{th}}p\otimes p\otimes\cdots\in \mathcal A.
\end{equation}
We define a subset of projections in $\mathcal A$ by
\begin{equation}\label{eq:projection_space}
\mathcal P_0:=\{p_{[n}:\,p\in \mathcal B, \text{ a
projection},\,\,n\in \mathbb Z_+\}.
\end{equation}
 \begin{define}\label{def4.5}
A quantum Markov chain is called reducible if there exists a
nontrivial projection $p\in \mathcal B $ and $n_0\in \mathbb Z_+$
 such
that \begin{equation}\label{eq:reducing}
E_{0]}(p_{[n_0}ap_{[n_0})=E_{0]}(a)
\end{equation}
 for all $a\in\mathcal{A}$. Otherwise it is called
irreducible. Any projection satisfying \eqref{eq:reducing} is called a reducing projection.
\end{define}
\begin{rem}
In the references \cite{AK1, AK2}, the reducing projections are
allowed to take much more general form. But here we will confine them to be of the forms in \eqref{eq:projection_space}. It will be turned out that this is enough.  
 \end{rem}
\begin{thm}\label{thm4.5-1}
The QMC is reducible if and only if $E_{0]}(I-p_{[n_0})=0$ for some
nontrivial projection $p_{[n_0}$.
\end{thm}
\begin{proof}
In the proof, for notational simplicity we just put $p$ for
$p_{[n_0}$. Suppose that $p$ is a nontrivial projection  such that
$E_{0]}(I-p)=0$. That is, $E_{0]}(p^\bot)=0$. Since $E_{0]}$ is
completely positive, it satisfies a Schwarz inequality:
$E_{0]}(b)^*E_{0]}(b)\le E_{0]}(b^*b)$ (see Theorem 2.10 of
\cite{F}, for example). Therefore,
$$E_{0]}(pap^\bot)^*E_{0]}(pap^\bot)\le E_{0]}(p^\bot a^*pap^\bot)\le E_{0]}(p^\bot a^*ap^\bot)\le \|a\|^2E_{0]}(p^\bot)=0.$$
Thus $E_{0]}(pap^\bot)=0$ and so $E_{0]}(p^\bot ap)=0$. Similarly we have  $E_{0]}(p^\bot ap^\bot)=0$. Therefore we get
\[
E_{0]}(a)=E_{0]}((p+p^\bot)a(p+p^\bot)) =E_{0]}(pap),
\]
for all $a\in\mathcal{A}$. This means that the QMC is reducible. The converse trivially holds by taking $a=I$.
\end{proof}

\section{Quantum  Markov  chains associated with OQRWs}

In this section, we construct QMCs associated with
OQRWs. As mentioned in the Introduction, this is a slight modification of the one developed in \cite{DM}. We will construct a nonhomogeneous QMC, but in \cite{DM}, a homogeneous QMC was considered. We will use notations from the previous section. In the sequel, we also use the density matrices as also for states (positive definite functions, in general), i.e., if $\rho$ is a positive definite trace class operator in $\mathcal B$, then for any $a\in \mathcal B$, we write $\mathrm{Tr}(\rho a)$ or $\rho(a)$ denoting the same value of the functional at $a$. Let us define some notations which will be used in the sequel. For $i, j\in\Lambda$, a path from
$i$ to $j$ is any finite sequence $i_0, i_1, \cdots, i_l$ in
$\Lambda$ with $l\ge1$, such that $i_0=i$ and $i_l=j$. We denote
such a path by $\pi(i_0,\cdots,i_l)$ and let $\mathcal{P}(i,j)$ be
the set of all paths from $i$  to $j$. For $\pi(i_0,\cdots,i_l)$
in $\mathcal{P}(i,j)$ we denote by $B_{\pi(i_0,\cdots,i_l)}$ the
operator on $\mathcal{H}$:
$$B_{\pi(i_0,\cdots,i_l)}=B_{i_{l-1}}^{i_l}\cdots B_{i_0}^{i_1}=B_{i_{l-1}}^{j}\cdots B_{i}^{i_1}.$$

\subsection{QMCs for OQRWs}

Let $\mathcal{M}$ be an OQRW given by (\ref{2.3}). We fix a
density operator
$\rho^{(0)}\in\mathcal{B}(\mathcal{H}\otimes\mathcal{K})$ of the
form
$$\rho^{(0)}=\sum_i\rho^{(0)}_i\otimes|i\rangle\langle i|,$$
 where $\rho_i^{(0)}\ge0$ and $\sum_i\text{Tr}(\rho^{(0)}_i)=1$  for all $i$.
For an initial state $\rho^{(0)}$,
$\rho^{(n)}:=\mathcal{M}^n(\rho^{(0)})$ is the state at time $n$.
Then we can write
\begin{equation}\label{4.1}
\rho^{(n)}=\sum_i\rho^{(n)}_i\otimes|i\rangle\langle i|.
\end{equation}
We would like to remind the reader that starting with any initial state, even not of the block-diagonal form, after the evolution of OQRW the states result in the block-diagonal form as in \eqref{4.1} \cite{APSS}. Therefore, it is natural and sufficient to consider also the observables of the block-diagonal form. So, define a subalgebra $\mathcal B_0\subset \mathcal B(\mathcal H\otimes \mathcal K)$ by 
\begin{equation}\label{eq:subalgebra}
\mathcal B_0=\{\sum_{i\in \Lambda}a(i)\otimes |i\rangle\langle i|:a(i)\in \mathcal B(\mathcal H)\text{ for all }i\in \Lambda\text{ and }\sum_i\|a(i)\|<\infty\}.
\end{equation}
Let $\mathcal B$ be the von Neumann subalgebra of $\mathcal B(\mathcal H\otimes \mathcal K)$ obtained by a weak closure of $\mathcal B_0$. 
We consider the algebra
$$\mathcal{A}=\bigotimes_{i\in\mathbb{Z}_+}\mathcal{A}_i$$
where $\mathcal{A}_i=\mathcal{B}$
for all $i\in\mathbb{Z}_+$. For each $n=0, 1,2,\cdots$, define the
following operators
\begin{eqnarray}
&&A^{(n)}_{ij}=\frac{1}{\text{Tr}(\rho^{(n)}_j)^{1/2}}((\rho_j^{(n)})^{1/2}\otimes|i\rangle\langle j|),
\quad i, j\in\Lambda,\nonumber\\
&&K_{ij}^{(n)}={M^i_j}^*\otimes A_{ij}^{(n)}.\label{4.2}
 \end{eqnarray}
 Here it is assumed $A^{(n)}_{ij}=0$ if $\rho^{(n)}_j=0$. Notice that by this convention, we can allow any kind of initial states $\rho^{(0)}$ so that $\rho^{(0)}_i$ might be zero for some $i\in \Lambda$. This is important when we recover the dynamics of OQRW itself by the QMC. See Proposition \ref{prop:recovering_classical_property}.
 \begin{prop}\label{prop4.1}
 For each $n=0, 1,2,\cdots$,
$$\mathrm{Tr}_2(\sum_{i,j}K_{ij}^{(n)}{K_{ij}^{(n)}}^*)\le I$$ holds.
\end{prop}
\begin{proof}
\begin{eqnarray*}
\text{Tr}_2(\sum_{i,j}K_{ij}^{(n)}{K_{ij}^{(n)}}^*)&=&
{\sum_{j:\rho_j^{(n)}\neq 0}}\sum_i\frac{\text{Tr}(\rho^{(n)}_j\otimes|i\rangle\langle
i|)}{\text{Tr}(\rho^{(n)}_j)}{M^i_j}^*M^i_j\\
&=&{\sum_{j:\rho_j^{(n)}\neq 0}}\sum_{i}{B^i_j}^*B^i_j\otimes |j\rangle\langle j|\\
&=&{\sum_{j:\rho_j^{(n)}\neq 0}}I_{\mathcal H}\otimes |j\rangle\langle j|\le I.
\end{eqnarray*}
This proves the assertion.
\end{proof}
By the above proposition we can define transition expectations.
 \begin{define}\label{def4.2}(Transition expectations)
For each $n=0,1,2,\cdots$, and $x,y\in \mathcal B$, define
\begin{eqnarray}
\mathcal{E}^{(n)}(x\otimes y)&:=&\sum_{i,j}\text{Tr}_2(K_{ij}^{(n)}(y\otimes x){K_{ij}^{(n)}}^*)\nonumber\\
                             &=&{\sum_{j:\rho_j^{(n)}\neq 0}}\sum_{i}\frac{\text{Tr}(\rho_j^{(n)}\otimes|j\rangle\langle
                             j|x)}{\text{Tr}(\rho_j^{(n)})}{M^i_j}^*yM^i_j.\label{4.3}
  \end{eqnarray}
 \end{define}
 The above transition expectations are of the form in \eqref{3.4}, but before taking a partial trace a transposition was applied, leading to the  transpose transition expectation $\mathcal{E}^t$ of \cite{DM}. To say more, one may construct transition expectations by changing the roles of $x$ and $y$ in \eqref{4.3}, which gives rise to define a new QMC. But it turns out that the present form is very convenient when we  talk about the dynamics of OQRWs. See, e.g., Proposition \ref{prop:recovering_classical_property}.  
Using the above transition expectations, we define the completely
positive maps $E_{m]}:\mathcal A\to \mathcal A_{m]}$ by
\eqref{eq:cp_conditional_expectation} and define a positive definite functional $\rho$ on
$\mathcal A$ like in  \eqref{eq:qmc}:
\begin{equation}\label{eq:state_qmc}
\rho(a):=\rho^{(0)}(E_{0]}(a)),\quad a\in \mathcal A.
\end{equation} 
Before going further, we refine Lemma \ref{lem:cp-sub-Markovian map} for the present model by showing the following property. Recall the definition given in Lemma \ref{lem:cp-sub-Markovian map}:
\[
\ol b(n)=\lim_{k\to \infty}\mathcal E^{(n)}(I\otimes \mathcal E^{(n+1)}(I\otimes \cdots\otimes\mathcal E^{(n+k)}(I\otimes I))).
\]
For a state of the form $\rho=\sum_i\rho_i\otimes |i\rangle\langle i|$, we let $\L(\rho):=\{i\in \L:\rho_i\neq 0\}$.
\begin{lem}\label{eq:limit_operators_properties}
The operators $\{\ol b(n)\}_{n\ge 0}$ 
for the transition expectations of OQRWs satisfy the following properties.
\begin{enumerate}
\item[(i)] For each $n\ge 0$ and $j\in \L({\rho^{(n)}})$, there exist strictly positive operators $\ol b(n,j)\in \mathcal B(\mathcal H)$ such that 
\[
\ol b(n)=\sum_{j\in \L({\rho^{(n)}})}\ol b(n,j)\otimes |j\rangle\langle j|.
\]
\item[(ii)] For each $j\in \L({\rho^{(n)}})$, it holds that
\[
\sum_{i\in\L(\rho^{(n+1)})}{B_j^i}^*\ol b(n+1,i)B_j^i=\ol b(n,j). 
\]
\item[(iii)] For each $n\ge 0$ and $j\in \L({\rho^{(n)}})$,
\[
\mathrm{Tr}(\rho^{(n)}_j\ol b(n,j))=\mathrm{Tr}(\rho_j^{(n)}).
\]
\end{enumerate}
\end{lem}
\begin{proof}
(i) Define $a^{(n)}_k:=\mathcal E^{(n)}(I\otimes \mathcal E^{(n+1)}(I\otimes \cdots\otimes\mathcal E^{(n+k)}(I\otimes I)))$. Then we have $\ol b(n)=\lim_{k\to \infty}a_k^{(n)}$. By directly computing with the definition \eqref{4.3} we get 
\[
a_k^{(n)}= \sum_{i_{n}\in \Lambda(\rho^{(n)})}b^{(n)}(i_n;k)\otimes|i_{n}\rangle\langle i_{n}|,
\] 
where 
\[
b^{(n)}(i_n;k)= \sum_{i_{n+1}\in\Lambda(\rho^{(n+1)})}\cdots\sum_{i_{n+k}\in \Lambda(\rho^{(n+k)})}\,B^*_{\p(i_{n},\cdots, i_{n+k})}B_{\p(i_{n},\cdots, i_{n+k})}.
\]  
By the property \eqref{2.1} we see that $\{b^{(n)}(i_n;k)\}_{k\ge 1}$ is a sequence of decreasing positive definite operators on $\mathcal B$. Thus by Vigier's Theorem \cite{RSz.-N} again, we see that the sequence converges strongly to a nonnegative element, say $\ol b(n,i_n)$ as $k\to \infty$.
We thus get 
\[
\ol b(n)=\lim_{k\to \infty}a_k^{(n)}=\sum_{i_n\in \Lambda(\rho^{(n)})}\,\ol b(n,i_n)\otimes |i_n\rangle\langle i_n|.
\]
The strict positivity of $\ol b(n,j)$ for $j\in \L(\rho^{(n)})$ follows from (iii) whose proof does not use this property.\\
(ii) By the computations in (i), we see that for $j\in \L(\rho^{(n)})$,
\begin{eqnarray*}
&&\sum_{i_{n+1}\in \Lambda(\rho^{(n+1)})}{B_j^{i_{n+1}}}^*\,\ol b(n+1,i_{n+1})B_j^{i_{n+1}}\\
&=& \lim_{k\to \infty} \sum_{i_{n+1}\in\Lambda(\rho^{(n+1)})}{B_j^{i_{n+1}}}^*\,\Big(\sum_{i_{n+2}\in\Lambda(\rho^{(n+2)})}\cdots\sum_{i_{n+k}\in
\Lambda(\rho^{(n+k)})}\,B^*_{\p(i_{n+1},\cdots, i_{n+k})}B_{\p(i_{n+1},\cdots, i_{n+k})}\Big)B_j^{i_{n+1}}\\
&=&\lim_{k\to \infty}\sum_{i_{n+1}\in\Lambda(\rho^{(n+1)})}\cdots\sum_{i_{n+k}\in
\Lambda(\rho^{(n+k)})}\,B^*_{\p(j,i_{n+1},\cdots, i_{n+k})}B_{\p(j,i_{n+1},\cdots, i_{n+k})}\\
&=&\ol b(n,j).
\end{eqnarray*}
(iii) We see again
\begin{eqnarray*}
&&\mathrm{Tr}(\rho^{(n)}_j\ol b(n,j))\\
&=& \lim_{k\to \infty}\sum_{i_{n+1}\in\Lambda(\rho^{(n+1)})}\cdots\sum_{i_{n+k}\in
\Lambda(\rho^{(n+k)})}\, \mathrm{Tr}\left(\rho^{(n)}_jB^*_{\p(j,i_{n+1},\cdots, i_{n+k})}B_{\p(j,i_{n+1},\cdots, i_{n+k})}\right) \\
&=& \lim_{k\to \infty}{\sum_{i_{n+1}}}\cdots{\sum_{i_{n+k}}}\, \mathrm{Tr}\left(\rho^{(n)}_jB^*_{\p(j,i_{n+1},\cdots, i_{n+k})}B_{\p(j,i_{n+1},\cdots, i_{n+k})}\right) \\
&=&\lim_{k\to \infty}\mathrm{Tr}(\rho_j^{(n)})\\
&=&\mathrm{Tr}(\rho_j^{(n)}).
\end{eqnarray*}
Here in the third equality the relation \eqref{2.1} was used and the second equality can be shown by the following argument. Suppose, for example, $i_{n+l}\notin\Lambda(\rho^{(n+l)}) $ for some $1\le l\le k$. We claim that 
\[
\mathrm{Tr}\left(\rho^{(n)}_jB^*_{\p(j,i_{n+1},\cdots,i_{n+l},\cdots, i_{n+k})}B_{\p(j,i_{n+1},\cdots, i_{n+l},\cdots,i_{n+k})}\right)=0.
\]
In fact, 
\begin{eqnarray*}
&&\mathrm{Tr}\left(\rho^{(n)}_jB^*_{\p(j,i_{n+1},\cdots,i_{n+l},\cdots, i_{n+k})}B_{\p(j,i_{n+1},\cdots, i_{n+l},\cdots,i_{n+k})}\right)\\
&=&\mathrm{Tr}\left(B_{\p(j,i_{n+1},\cdots, i_{n+l},\cdots,i_{n+k})}\rho^{(n)}_jB^*_{\p(j,i_{n+1},\cdots,i_{n+l},\cdots, i_{n+k})}\right)\\
&=&\mathrm{Tr}\left(B_{\p(i_{n+l},\cdots,i_{n+k})}B_{\p(j,i_{n+1},\cdots, i_{n+l})}\rho^{(n)}_jB^*_{\p(j,i_{n+1},\cdots,i_{n+l})}B^*_{\p(i_{n+l},\cdots, i_{n+k})}\right).
\end{eqnarray*}
But using the definition of OQRW in \eqref{2.4} we see that 
\begin{eqnarray*}
0&\le&B_{\p(j,i_{n+1},\cdots, i_{n+l})}\rho^{(n)}_jB^*_{\p(j,i_{n+1},\cdots,i_{n+l})}\\
&\le&\sum_j\sum_{i_{n+1}}\cdots\sum_{i_{n+l-1}}B_{\p(j,i_{n+1},\cdots, i_{n+l})}\rho^{(n)}_jB^*_{\p(j,i_{n+1},\cdots,i_{n+l})}\\
&=&\rho^{(n+l)}_{i_{n+l}}=0,
\end{eqnarray*}
by the assumption that $i_{n+l}\notin\Lambda(\rho^{(n+l)})$. This proves the claim and the proof is completed.
 \end{proof} 
In the sequel, by abuse of notations and to save the space, we use ${\sum'_{i_l}}$ for $\sum_{i_l\in \Lambda(\rho^{(l)})}$ whenever there is no danger of confusion.
\begin{lem}\label{lem:reducing_computation}
For any $a_{n]}=a_0\otimes\cdots\otimes a_n\otimes I_{[n+1}\in \mathcal A_{n]}$, we have 
\begin{eqnarray}\label{eq:reducing_computation}
&&E_{0]}(a_{n]})\\
&=&{\sum_{i_0}}'\cdots {\sum_{i_n}}'\,\prod_{k=0}^n\frac{\mathrm{Tr}(\rho_{i_k}^{(k)}\otimes |i_k\rangle\langle i_k|a_k)}{\mathrm{Tr}(\rho_{i_k}^{(k)})}\left(B^*_{\p(i_0,\cdots,i_n)}\ol b(n,i_n)B_{\p(i_0,\cdots,i_n)}\otimes |i_0\rangle\langle i_0|\right).\nonumber
\end{eqnarray} 
\end{lem}
\begin{proof}
Recall 
\[
E_{0]}(a_{n]})=\mathcal E^{(0)}(a_0\otimes\cdots \otimes\mathcal E^{(n)}(a_n\otimes \ol b(n+1))).
\]
By definition \eqref{4.3} and Lemma \ref{eq:limit_operators_properties} (i) and (ii), we see that
\begin{eqnarray*}
&&\mathcal E^{(n)}(a_n\otimes \ol b(n+1))\\
&=&{\sum_{i_n}}'\,\frac{\mathrm{Tr}(\rho_{i_n}^{(n)}\otimes |i_n\rangle\langle i_n|a_n)}{\mathrm{Tr}(\rho_{i_n}^{(n)})} {\sum_{i_{n+1}}}'\,{B_{i_n}^{i_{n+1}}}^*\,\ol b(n+1,i_{n+1})B_{i_n}^{i_{n+1}}\otimes |i_n\rangle\langle i_n|\\
&=&{\sum_{i_n}}'\,\frac{\mathrm{Tr}(\rho_{i_n}^{(n)}\otimes |i_n\rangle\langle i_n|a_n)}{\mathrm{Tr}(\rho_{i_n}^{(n)})}\ol b(n,i_n)\otimes |i_n\rangle\langle i_n|.
\end{eqnarray*}
Now repeated application of \eqref{4.3} and Lemma \ref{eq:limit_operators_properties} (i) gives the result.
 \end{proof} 
The following proposition shows two important features of our definition. One is that for any initial state $\rho^{(0)}$, the pair          $(\rho^{(0)},(\mathcal E^{(n)})_{n\ge 0})$ is a Markov pair (see Corollary \ref{cor:Markovian}), in other words, $\rho$ in \eqref{eq:state_qmc} is a state on $\mathcal A$ and hence a QMC. The second one is that the QMCs associated with OQRWs naturally extend the classical Markov chains (see \eqref{eq:state} in Subsection \ref{subsec:CMC}).
\begin{prop}\label{prop:recovering_classical_property}
For any $x\in \mathcal B$,
\[
\rho(I\otimes  \cdots \otimes I\otimes \overset{n\mathrm{th}}{x}\otimes I\otimes \cdots)=\rho^{(n)}(x),
\]
where $\rho^{(n)}=\mathcal M^n(\rho^{(0)})$.
\end{prop}
\begin{proof}
Using the definition \eqref{3.2}, we get 
\begin{eqnarray*}
&&\rho(I\otimes  \cdots \otimes I\otimes \overset{n\mathrm{th}}{x}\otimes I\otimes \cdots)\\
&=& \rho^{(0)}\left(\mathcal{E}^{(0)}
(I\otimes\cdots \otimes\mathcal{E}^{(n)}(x\otimes \ol b(n+1)))  \right).
\end{eqnarray*}
By Lemma \ref{lem:reducing_computation} and Lemma \ref{eq:limit_operators_properties} (iii), we have 
\begin{eqnarray*}
&&\rho^{(0)}\left(\mathcal{E}^{(0)}
(I\otimes\cdots \otimes\mathcal{E}^{(n)}(x\otimes \ol b(n+1)))  \right)\\
&=&  {\sum_{i_n}}'\,\frac{\mathrm{Tr}(\rho_{i_n}^{(n)}\otimes |i_n\rangle\langle i_n|x)}{\mathrm{Tr}(\rho_{i_n}^{(n)})}{\sum_{i_{0}}}'\cdots{\sum_{i_{n-1}}}'   \,\mathrm{Tr}\left(B_{\p(i_0,\cdots,i_n)}\rho_{i_0}^{(0)}{B^*_{\p(i_0,\cdots,i_n)}}\ol b(n,i_n) \right)\\
&=&  {\sum_{i_n}}'\,\frac{\mathrm{Tr}(\rho_{i_n}^{(n)}\otimes |i_n\rangle\langle i_n|x)}{\mathrm{Tr}(\rho_{i_n}^{(n)})}{\sum_{i_{0}}}\cdots{\sum_{i_{n-1}}}   \,\mathrm{Tr}\left(B_{\p(i_0,\cdots,i_n)}\rho_{i_0}^{(0)}{B^*_{\p(i_0,\cdots,i_n)}}\ol b(n,i_n) \right)\\
&=&  {\sum_{i_n}}'\,\frac{\mathrm{Tr}(\rho_{i_n}^{(n)}\otimes |i_n\rangle\langle i_n|x)}{\mathrm{Tr}(\rho_{i_n}^{(n)})}   \,\mathrm{Tr}\left( \rho_{i_n}^{(n)}\ol b(n,i_n) \right)\\
&=&  {\sum_{i_n}}'\,\frac{\mathrm{Tr}(\rho_{i_n}^{(n)}\otimes |i_n\rangle\langle i_n|x)}{\mathrm{Tr}(\rho_{i_n}^{(n)})}   \,\mathrm{Tr} ( \rho_{i_n}^{(n)}   )\\
 &=&{\sum_{i_n}}'\mathrm{Tr}(\rho_{i_n}^{(n)}\otimes |i_n\rangle\langle i_n|x)\\
&=&\rho^{(n)}(x).
\end{eqnarray*} 
The second and third equalities follow from the definition of OQRWs. The proof is completed.
\end{proof}
\begin{cor}\label{cor:Markovian}
The pair $(\rho^{(0)},(\mathcal E^{(n)})_{n\ge 0})$ is a Markov pair.
\end{cor}
\begin{proof} It follows from Proposition \ref{prop:recovering_classical_property} by taking $x=I$.
\end{proof}
\begin{define}\label{def:qmc_oqrw}
The pair $(\rho^{(0)},(\mathcal E^{(n)})_{n\ge 0})$, or the state
$\rho$ in \eqref{eq:state_qmc} is called the (nonhomogeneous) QMC associated with
the OQRW.
\end{define}
We remark that as will be noted in Subsection \ref{subsec:CMC},
the property in Proposition
\ref{prop:recovering_classical_property} is observed when the QMCs
are applied to recover the classical Markov chains (see \eqref{eq:state}), and this property was already observed in \cite{AK1, AK2}.

Next we shortly discuss the invariant states for the QMCs.
 \begin{define}\label{def4.3}(Invariant state)
A state (density matrix) $\omega $ on $\mathcal B $ is called invariant to the QMC if
$$\mathrm{Tr}(\omega x)=\text{Tr}(\omega \mathcal{E}^{(n)}(I\otimes x))$$
for all $x$ and $n=0,1,2,\cdots$.
 \end{define}
 This corresponds to the condition (2.3) of \cite{P}. The following proposition shows that an invariant state
 $\omega $ to the Markov
 chain of a OQRW is an invariant state (density operator) with respect to
 $\mathcal{M}$.
\begin{prop}\label{prop4.4} A state $\omega =\sum_i\omega_i\otimes |i\rangle\langle i|$ is invariant to the QMC of OQRW if and only if
$\sum_{i,j}\mathrm{Tr}_2({K_{ij}^{(n)}}^*\omega $ $\otimes I K_{ij}^{(n)})=\omega $ for all $n\ge 0$,
and in this case $\omega$ satisfies $\omega =\sum_{i,j}M^i_j\omega {M^i_j}^*$, that is, $\omega =\mathcal{M}(\omega )$. On the other hand, if $\omega =\mathcal{M}(\omega )$, the state $\omega$ is invariant to the QMC $(\rho^{(0)},(\mathcal E^{(n)})_{n\ge 0})$ with $\rho^{(0)}=\omega$. In this case we have $\mathcal E^{(n)}=\mathcal E^{(0)}$ for all $n\ge 0$, i.e., the QMC is homogeneous.
\end{prop}
\begin{proof} We have
\begin{eqnarray*}
\text{Tr}(\omega \mathcal{E}^{(n)}(I\otimes x))&=& \sum_{i,j}\text{Tr}\big(\text{Tr}_2((\omega \otimes I)(K_{ij}^{(n)}(x\otimes I){K_{ij}^{(n)}}^*))\big)\\
&=&\sum_{i,j}\tilde{\text{Tr}}\big((\omega \otimes I)(K_{ij}^{(n)}(x\otimes I){K_{ij}^{(n)}}^*)\big)\\
&=&\sum_{i,j}\text{Tr}\big(\text{Tr}_2\big({K_{ij}^{(n)}}^*(\omega \otimes I)K_{ij}^{(n)}(x\otimes I)\big)\big)\\
&=&\sum_{i,j}\text{Tr}\big(\text{Tr}_2\big({K_{ij}^{(n)}}^*(\omega \otimes
I)K_{ij}^{(n)}\big)x\big).
\end{eqnarray*}
Thus $\text{Tr}(\omega \mathcal{E}^{(n)}(I\otimes x))=\text{Tr}(\omega x)$ for
all $x$ if and only if
$\sum_{i,j}\text{Tr}_2({K_{ij}^{(n)}}^*\omega \otimes I
K_{ij}^{(n)})=\omega $. By direct calculation, we have
\begin{eqnarray*}
\sum_{i,j}\text{Tr}_2({K_{ij}^{(n)}}^*\omega \otimes I K_{ij}^{(n)})&=&\sum_{j:\rho_j^{(n)}\neq 0}\sum_{i}M^i_j\omega {M^i_j}^*\\
&=&\sum_{j:\rho_j^{(n)}\neq 0}\sum_{i}B_j^i\omega_j{B_j^i}^*\otimes |i\rangle\langle i|.
\end{eqnarray*}
Therefore $\sum_{i,j}\text{Tr}_2({K_{ij}^{(n)}}^*\omega \otimes I
K_{ij}^{(n)})=\omega $ if and only if 
\begin{equation}\label{eq:support}
\sum_i\omega_i\otimes |i\rangle\langle i|=\sum_{i}\left(\sum_{j:\rho_j^{(n)}\neq 0}B_j^i\omega_j{B_j^i}^*\right)\otimes |i\rangle\langle i|.
\end{equation}
By taking trace to both sides of the above equation we get
\begin{eqnarray*}
1&=&\sum_i\sum_{j:\rho_j^{(n)}\neq 0}\mathrm{Tr}\left(B_j^i\omega_j{B_j^i}^*\right)\\
&=&\sum_{j:\rho_j^{(n)}\neq 0} \mathrm{Tr}(\omega_j).
\end{eqnarray*}
This means that $\omega_i=0$ if $\rho_i^{(n)}=0$ (for all $n\ge 0$). Thus  \eqref{eq:support} is written as 
\[
\sum_i\omega_i\otimes |i\rangle\langle i|=\sum_i\sum_j B_j^i\omega_j{B_j^i}^*\otimes |i\rangle\langle i|=\sum_i\mathcal M(\omega)_i\otimes |i\rangle\langle i|.
\]
We have therefore $\omega=\mathcal M(\omega)$. 

Now conversely suppose $\omega=\mathcal M(\omega)$ and define a Markov pair $(\rho^{(0)},(\mathcal E^{(n)})_{n\ge 0})$ with $\rho^{(0)}=\omega$. Then, since $\rho^{(n)}=\rho^{(0)}=\omega$ for all $n\ge 0$, it is a homogeneous QMC. Moreover, by \eqref{4.3}
\begin{eqnarray*}
\mathrm{Tr}(\omega \mathcal E^{(n)}(I\otimes x))&=&\sum_{j:\omega_j\neq 0}\sum_i\mathrm{Tr}(\omega {M_j^i}^*x{M_j^i})\\
&=&\sum_{j:\omega_j\neq 0}\sum_i\mathrm{Tr}(B_j^i\omega_j{B_j^i}^*\otimes |i\rangle\langle i|x)\\
&=&\sum_{j}\sum_i\mathrm{Tr}(B_j^i\omega_j{B_j^i}^*\otimes |i\rangle\langle i|x)\\
&=&\mathrm{Tr}(\mathcal M(\omega)x)=\mathrm{Tr}(\omega x). 
\end{eqnarray*}
Therefore, $\omega$ is invariant to the QMC $(\rho^{(0)},(\mathcal E^{(n)})_{n\ge 0})$ with $\rho^{(0)}=\omega$.
\end{proof}

\subsection{Reducibility and irreducibility of QMCs for OQRWs}

Recall the definition of reducibility and irreducibility of QMCs
in Definition \ref{def4.5} with the projections in
\eqref{eq:projection_prototype} and \eqref{eq:projection_space}.
When we consider the reducibility and irreducibility problems for
QMCs associated with OQRWs, the possible reducing projections shall be of the form:
\begin{equation}\label{eq:projections_for_oqrw}
p_{[n}=I\otimes \cdots \otimes I\otimes \overset{n\text{th}}p\otimes p\otimes\cdots  \in \mathcal A\text{ with }
 p=\sum_jp(j)\otimes |j\rangle\langle j|\in \mathcal B,
\end{equation}
where $p(j)$'s are projections on $\mathcal H$. Then we define
\begin{equation}\label{eq:projection_space_for_oqrw}
\mathcal P_0:=\{p_{[n}:\,p_{[n}, \,\text{ a projection of the form
\eqref{eq:projections_for_oqrw}},\,\, n\in \mathbb Z_+\}.
\end{equation}
We say that a QMC associated with an OQRW is reducible if there exists a nontrivial projection $p_{[n_0}\in \mathcal P_0$ satisfying \eqref{eq:reducing} in Definition \ref{def4.5}. Otherwise it is called irreducible.

Let $0\le m\le n$ and $0\le n_0\le n$. Consider $a=a_0\otimes a_1\otimes\cdots\otimes a_m\otimes I_{[m+1}\in\mathcal{A}_{m]}$
and a projection $p_{[n_0,n]}:=I\otimes\cdots \otimes I\otimes \overset{n_0\text{th}}p\otimes\cdots p\otimes I_{[n+1}$. Notice that $p_{[n_0}=\lim_{n\to \infty}p_{[n_0,n]}$. In order to compute $E_{0]}(p_{[n_0,n]}ap_{[n_0,n]})$, we let for the time being
\begin{equation}\label{eq:projection_labeling}
p_l:=\begin{cases} I,&0\le l\le n_0-1,\\
p,&n_0\le l\le n. \end{cases}
\end{equation}
By Lemma \ref{lem:reducing_computation} we get for $a=a_0\otimes a_1\otimes\cdots\otimes a_m\otimes I_{[m+1}$, $m\le n$,
\begin{eqnarray}\label{4.5}
&&E_{0]}(p_{[n_0,n]}ap_{[n_0,n]})\\
&&={\sum_{i_0}}'\cdots{\sum_{i_{n}}}'\prod_{k=0}^m\frac{\text{Tr}(\rho_{i_{k}}^{(k)}p_k(i_k)a_k(i_k)p_k(i_k))}{\text{Tr}(\rho_{i_{k}}^{(k)})}
\prod_{k=m+1}^n\frac{\text{Tr}(\rho_{i_{k}}^{(k)}p_k(i_{k}))}{\text{Tr}(\rho_{i_{k}}^{(k)})}\nonumber\\
&&\times \left(B^*_{\p(i_0,\cdots,i_{n})} \ol b(n,i_n) B_{\p(i_0,\cdots,i_{n})}\otimes|i_0\rangle\langle i_0|\right),\nonumber
\end{eqnarray}
where $p_k$'s are given by \eqref{eq:projection_labeling}.
In particular, we have
\begin{eqnarray}\label{4.6}
 &&E_{0]}(p_{[n_0,n]})\\
 &&={\sum_{i_0}}'\cdots{\sum_{i_{n}}}' 
\prod_{k=n_0}^n\frac{\text{Tr}(\rho_{i_{k}}^{(k)}p_k(i_{k}))}{\text{Tr}(\rho_{i_{k}}^{(k)})} \left(B^*_{\p(i_0,\cdots,i_{n})} \ol b(n,i_n) B_{\p(i_0,\cdots,i_{n})}\otimes|i_0\rangle\langle i_0|\right).\nonumber
\end{eqnarray}
\begin{lem}\label{lem4.6}
For $p_{[n_0}\in \mathcal P_0$, one has
\begin{equation}\label{4.7}
I-p_{[n_0}=\sum_{n\ge n_0}^\infty I\otimes \cdots I\otimes \overset{n_0\text{th}}p\otimes\cdots\otimes p\otimes \overset{n\text{th}}{p^\perp}\otimes I_{[n+1}
\end{equation}
where $p^\bot=I-p$.
\end{lem}
\begin{proof}
Let us adopt the notations in \eqref{eq:projection_labeling}. We have
$$p_0^\bot\otimes I_{[1}+p_0\otimes p_1^\bot\otimes I_{[2}=I-p_0\otimes p_1\otimes I_{[2}.$$
Continuing this procedure, we have
$$\sum_{k\ge0}^n p_0\otimes p_1\otimes\cdots\otimes p_{k-1}\otimes p_k^\bot\otimes I_{[k+1}
=I-p_0\otimes p_1\otimes\cdots\otimes  p_n\otimes I_{[n+1}.$$
So taking the limit $n\rightarrow\infty$, and returning back the notations, we get (\ref{4.7}).
\end{proof}
\begin{prop}\label{prop4.7}
Let $p_{[n_0}\in \mathcal P_0$. Then,
$E_{0]}(I-p_{[n_0})=0$ if and only if $\rho^{(n)}_jp(j)=\rho^{(n)}_j$ for all $j\in \Lambda$ and $n\ge n_0$.
\end{prop}
\begin{proof}
If $E_{0]}(I-p_{[n_0})=0$ then by \eqref{4.7} we have
$$\sum_{n\ge n_0}^\infty E_{0]}(I\otimes \cdots I\otimes \overset{n_0\text{th}}p\otimes\cdots\otimes p\otimes \overset{n\text{th}}{p^\perp}\otimes I_{[n+1})=0.$$
Therefore, by Lemma \ref{lem:reducing_computation}, we have
for $n\ge n_0$,
\begin{eqnarray}\label{eq:vanishing}
&&E_{0]}(I\otimes \cdots I\otimes \overset{n_0\text{th}}p\otimes\cdots\otimes p\otimes \overset{n\text{th}}{p^\perp}\otimes I_{[n+1})\\
&&={\sum_{i_0}}'\cdots{\sum_{i_n}}'\prod_{k=n_0}^{n-1}
 \frac{\text{Tr}(\rho_{i_{k}}^{(k)}p(i_{k}))}{\text{Tr}(\rho_{i_{k}}^{(k)})}
 \frac{\text{Tr}(\rho_{i_{n}}^{(n)}p^\bot(i_{n}))}{\text{Tr}(\rho_{i_{n}}^{(n)})}
 \Big(B_{\pi(i_0,\cdots,i_{n})}^*\ol b(n,i_n) B_{\pi(i_0,\cdots,i_{n})}\nonumber\\
 &&\hskip 3true cm\otimes|i_0\rangle\langle i_0|\Big)=0.\nonumber
\end{eqnarray}
From this, we claim that  $\text{Tr}(\rho^{(n)}_jp(j)^\bot)=0$ for all $n\ge n_0$ and $j\in \Lambda(\rho^{(n)})$. In fact, first we see that
\begin{eqnarray*}
0&=&E_{0]}(I\otimes \cdots I\otimes \overset{n_0\text{th}}{p^\perp}\otimes   I_{[n_0+1})\\
&=& {\sum_{i_{n_0}}}'\, 
 \frac{\text{Tr}(\rho_{i_{n_0}}^{(n_0)}p^\bot(i_{n_0}))}{\text{Tr}(\rho_{i_{n_0}}^{(n_0)})}
{\sum_{i_0}}'\cdots{\sum_{i_{n_0-1}}}'\, \Big(B_{\pi(i_0,\cdots,i_{n_0})}^*\ol b(n_0,i_{n_0}) B_{\pi(i_0,\cdots,i_{n_0})} \otimes|i_0\rangle\langle i_0|\Big).
\end{eqnarray*}
Since 
\begin{eqnarray*}
&&\mathrm{Tr}\Big(\rho^{(0)}{\sum_{i_0}}'\cdots{\sum_{i_{n_0-1}}}'\, \Big(B_{\pi(i_0,\cdots,i_{n_0})}^*\ol b(n_0,i_{n_0}) B_{\pi(i_0,\cdots,i_{n_0})} \otimes|i_0\rangle\langle i_0|\Big)\Big)\\
&=&\mathrm{Tr}\left(\rho^{(n_0)}_{i_{n_0}}\ol b(n_0,i_{n_0})\right)\\
&=&\mathrm{Tr}\left(\rho^{(n_0)}_{i_{n_0}}\right)>0,
\end{eqnarray*}
which follows by Lemma \ref{lem:reducing_computation} (iii), the operator 
\[
{\sum_{i_0}}'\cdots{\sum_{i_{n_0-1}}}'\, \Big(B_{\pi(i_0,\cdots,i_{n_0})}^*\ol b(n_0,i_{n_0}) B_{\pi(i_0,\cdots,i_{n_0})} \otimes|i_0\rangle\langle i_0|\Big)
\]
is positive. Thus we conclude that $\text{Tr}(\rho^{(n_0)}_jp(j)^\bot)=0$ for $j\in \Lambda(\rho^{(n_0)})$. By induction and repeated use of \eqref{eq:vanishing} proves the claim. Now, since $\text{Tr}(\rho^{(n)}_jp(j)^\bot)=\text{Tr}(p(j)^\bot\rho^{(n)}_jp(j)^\bot)\ge0$, $p(j)^\bot\rho^{(n)}_jp(j)^\bot=0$ and so $\rho^{(n)}_jp(j)^\bot=0$, or $\rho^{(n)}_jp(j)=\rho^{(n)}_j$ for all $n\ge n_0$ and $j\in \Lambda(\rho^{(n)})$, and hence for all $j\in \Lambda$.

On the other hand, if $\rho^{(n)}_jp(j)=\rho^{(n)}_j$ for all $n\ge n_0$ and $j\in \Lambda$, we get from (\ref{4.6}) that
\[
 E_{0]}(I\otimes I\otimes\cdots)=E_{0]}(p_{[n_0,n]})={\sum_{i_0}}'\cdots{\sum_{i_{n}}}' 
  \left(B^*_{\p(i_0,\cdots,i_{n})} \ol b(n,i_n) B_{\p(i_0,\cdots,i_{n})}\otimes|i_0\rangle\langle i_0|\right). 
\]
Therefore, we have $E_{0]}(I-p_{[n_0,n]}) =0$. Taking the limit $n\to \infty$, we get $E_{0]}(I-p_{[n_0})=0$. 
The proof is completed.
\end{proof}
\begin{thm}\label{thm4.8}
The QMC associated with an OQRW is reducible with a reducing projection $p_{[n_0}\in \mathcal P_0$ if and only if  $\rho^{(n)}_jp(j)=\rho^{(n)}_j$
 for all $j\in \Lambda$ and $n\ge n_0$.
\end{thm}
\begin{proof}
The proof follows from Theorem \ref{thm4.5-1} and Proposition \ref{prop4.7}.
\end{proof}
\begin{thm}\label{thm4.9}
Suppose that $h$ is a nontrivial projection on $\mathcal{H}$ such
that 
$$hB_\pi=B_\pi$$
for any path $\pi\in\mathcal{P}(i,j)$ for all $i,j\in \Lambda$. Then the QMC is reducible.
\end{thm}
\begin{proof}
Define a projection $p \in \mathcal B(\mathcal H\otimes \mathcal K)$ by
 $p:=h\otimes I_{\mathcal K}=\sum_{j\in \Lambda}h\otimes|j\rangle\langle j|$ and consider $p_{[1}\in \mathcal P_0$. Then for $n\ge 1$,
\begin{eqnarray*}
 \rho^{(n)}_jh&=&\sum_{i_0\in \Lambda}\sum_{\pi\in\mathcal{P}(i_0,j)}B_\pi\rho_{i_0}B^*_\pi h\\
&=&\sum_{i_0\in \Lambda}\sum_{\pi\in\mathcal{P}(i_0,j)}B_\pi\rho_{i_0}B^*_\pi\\
&=&\rho^{(n)}_j.
\end{eqnarray*}
By Theorem \ref{thm4.8}, the QMC is reducible with a reducing projection $p_{[1}$.
\end{proof}
\begin{rem}\label{rem 4.10}
 (a) The condition $\rho^{(n)}_jp(j)=\rho^{(n)}_j$  for all $j\in \Lambda$ and $n\ge n_0$ in  Proposition \ref{prop4.7} is equivalent to
 $p(j)\rho^{(n)}_jp(j)=\rho^{(n)}_j$  for all $j\in \Lambda$ and $n\ge n_0$, which means that for each $j\in \Lambda$,
 the support of $\rho_j^{(n)}$ is in the range space of $p(j)$ for all $n\ge n_0$.

 (b) By  Theorem \ref{thm4.9}, if the range of $B_j^i$
 for all $i, j$ belongs to the nontrivial
 subspace, that is, $hB_j^i=B_j^i$ for some nontrivial projection $h$, then the QMC is reducible.
\end{rem}
Next we discuss some sufficient conditions for the irreducibility.
\begin{thm}\label{thm:irreducibility}
Suppose that the OQRW is such that $\rho^{(n)}_j/\mathrm{Tr}(\rho^{(n)}_j)$ is a faithful state on $\mathcal B(\mathcal H)$ for all $j\in \Lambda$ and $n=0,1,2,\cdots$. Then the QMC associated with this OQRW is irreducible.
\end{thm}
\begin{proof}
Suppose, on the contrary, that there is a nontrivial projection $p$ on $\mathcal H$ and $p_{[n_0}\in \mathcal P_0$ is a reducing projection  for the QMC.  Then by Theorem \ref{thm4.8} it follows that $\rho^{(n)}_jp(j)=\rho^{(n)}_j$ for all $j\in \Lambda$ and $n\ge n_0$. Since $\rho^{(n)}_j/\mathrm{Tr}(\rho^{(n)}_j)$ is a faithful state it must hold that $p(j)$ is the identity operator on $\mathcal H$, leading to a contradiction. 
\end{proof}
An example satisfying the conditions in the theorem will be considered in Subsection \ref{subsec:CMC}.
\begin{rem}
The reducibility and irreducibility of positive maps on the ideal of trace class operators (in Schr\"odinger representation), and equivalently, of positive maps on the operator algebras (in Heisenberg representation), was introduced in some literature, see for example, \cite{CP, FP}. Typically, the study of reducibility and irreducibility for OQRWs was investigated in \cite{CP}. It turns out that the concepts of reducibility and irreducibility for OQRWs defined in \cite{CP} and in this paper are equivalent. In the Appendix we will consider the equivalence in detail. Therefore, in particular, under the condition of Theorem \ref{thm4.9}, the OQRW is reducible in the sense of \cite{CP}. Also, under the condition of Theorem \ref{thm:irreducibility}, the OQRW is irreducible in the sense of \cite{CP}.  
\end{rem}

\section{Examples}
\subsection{OQRWs on the 1-dimensional integer lattice}
In this subsection we give some examples of reducible and irreducible OQRWs on the 1-dimensional integer lattice. Of course the idea can be extended to multi-dimensional models. First we consider reducible OQRWs.

\begin{ex}\label{Example 4.6}
Let us consider a stationary OQRW on
$\mathbb{Z}$ with nearest-neighbor jumps (see \cite{APSS}). Let
$\mathcal{H}$ be a Hilbert space and $B, C\in
\mathcal{B}(\mathcal{H})$ such that $B^*B+C^*C=I$. We define the
OQRW as follows: $$B^{i-1}_i=B~\text{and}~B^{i+1}_i=C$$ for all
$i\in\mathbb{Z}$, and $B^i_j=0$ in the other cases. Fix a density operator
$\rho\in \mathcal{B}(\mathcal{H}\otimes\mathcal{K})$, of the form
$$\rho=\sum_i\rho_i\otimes|i\rangle\langle i|$$
with $\rho_i\neq0$ for all $i$. We get
\begin{equation}\label{eq:1-d_oqrw}
\mathcal{M}(\rho)=\sum_j(B\rho_{j+1}B^*+C\rho_{j-1}C^*)\otimes|j\rangle\langle j|.
\end{equation}
In order to specify the model, let us consider the following matrices,
$$B=\left[\begin{matrix}0&0\\\frac{1}{\sqrt{2}}&\frac{1}{\sqrt{2}}\end{matrix}\right],
~~C=\left[\begin{matrix}0&0\\-\frac{1}{\sqrt{2}}&\frac{1}{\sqrt{2}}\end{matrix}\right],
~~h=\left[\begin{matrix}0&0\\0&1\end{matrix}\right],
$$
or
$$B=\left[\begin{matrix}\frac{1}{\sqrt{2}}&0\\-\frac{1}{\sqrt{2}}&0\end{matrix}\right],
~~C=\left[\begin{matrix}0&\frac{1}{\sqrt{2}}\\0&-\frac{1}{\sqrt{2}}\end{matrix}\right],
~~h=\left[\begin{matrix}\frac{1}{2}&-\frac{1}{2}\\-\frac{1}{2}&\frac{1}{2}\end{matrix}\right].
$$
For both cases, $B$ and $C$ satisfy $B^*B+C^*C=I$ and $hB=B,\,hC=C$. By Theorem
\ref{thm4.9}, the QMC corresponding to this OQRW is reducible.
\end{ex}
The following is an example of reducible OQRW in 1 dimension with 3 states.
\begin{ex}\label{Example 4.7}
 Let us consider a stationary OQRW on
$\mathbb{Z}$ with nearest-neighbor jumps. Let $\mathcal{H}$ be a
Hilbert space and $L_i\in \mathcal{B}(\mathcal{H}), i=1,2,3$, satisfy   $\sum_{i=1}^3L_i^*L_i=I$. We define the walk as follows:
$$B^{i-1}_i=L_1,~B^i_i=L_2~\text{and}~B^{i+1}_i=L_3$$ for
all $i\in\mathbb{Z}$, and $B^i_j=0$ for the other cases. The evolution \eqref{eq:1-d_oqrw} becomes now
\begin{equation*}
\mathcal{M}(\rho)=\sum_j(L_1\rho_{j+1}{L_1}^*+L_2\rho_{j}{L_2}^*+L_3\rho_{j-1}{L_3}^*)\otimes|j\rangle\langle
j|.
\end{equation*}
If we take the matrices
$$L_1=\left[\begin{matrix}0&0&0\\0&\frac{1}{\sqrt{2}}&\frac{1}{\sqrt{2}}\\0&0&0\end{matrix}\right],
~L_2=\left[\begin{matrix}0&0&0\\0&\frac{1}{\sqrt{2}}&-\frac{1}{\sqrt{2}}\\0&0&0\end{matrix}\right],
~L_3=\left[\begin{matrix}0&0&0\\0&0&0\\1&0&0\end{matrix}\right],
~h=\left[\begin{matrix}0&0&0\\0&1&0\\0&0&1\end{matrix}\right],
$$
it holds that
$hL_i=L_i,\,i=1,2,3$. Thus by Theorem \ref{thm4.9} again, the QMC is
reducible.
\end{ex}
Next we consider irreducible OQRWs in 1-dimensional space.
\begin{prop}\label{prop:irreducibility_1-dim}
In the 1-dimensional OQRW in \eqref{eq:1-d_oqrw}, suppose that $B$ and $C$ satisfy the following condition:
\begin{equation}\label{eq:irreducibility_sufficiency}
B^*xB=0\text{ and }C^*xC=0 \text{ for nonnegative } x\in \mathcal B(\mathcal H) \text{ implies } x=0.
\end{equation}
Then $\mathcal M(\rho)$ is faithful whenever $\rho$ is faithful. Therefore by Theorem \ref{thm:irreducibility}, the QMC associated with the OQRW \eqref{eq:1-d_oqrw} with a faithful initial state $\rho^{(0)}$ is irreducible.
\end{prop}
\begin{proof}
Let $\rho=\sum_{i\in \mathbb{Z}}\rho_i\otimes|i\rangle\langle i|$ be a faithful state. This means that $\rho_i$'s are faithful for all $i\in \mathbb Z$. We have to show that
\[
\mathcal M(\rho)_i=B\rho_{i+1}B^*+C\rho_{i-1}C^*
\]
is faithful for each $i\in \mathbb Z$. So, let $x\in \mathcal B(\mathcal H)$ be a nonnegative operator (matrix) and suppose that
\[
\Tr(\mathcal M(\rho)_i x)=\Tr((B\rho_{i+1}B^*+C\rho_{i-1}C^*)x)=0.
\]
Since $\rho_{i+1}$ as well as $\rho_{i-1}$ are faithful, it implies that
$B^*xB=0\text{ and }C^*xC=0$. By the condition \eqref{eq:irreducibility_sufficiency} we get $x=0$. The proof is completed.
\end{proof}
The simplest example for which the condition \eqref{eq:irreducibility_sufficiency} holds is the case where $B$ or $C$ is invertible. In the following example, the invertibility of $B$ or $C$ is not needed.
\begin{ex}
Let $U=\left[\begin{matrix}{\bf u}&{\bf v}\end{matrix}\right]$ be a $2\times 2$ unitary matrix with column vectors $\bf u$ and $\bf v$.
Let
\[
B=\left[\begin{matrix}{\bf u}&{0}\end{matrix}\right] \text{ and }C=\left[\begin{matrix}{0}&{\bf v}\end{matrix}\right].
\]
Then we get
\[
B^*xB=\left[\begin{matrix}\langle{\bf u},x{\bf u}\rangle&{0}\\0&0\end{matrix}\right] \text{ and }C^*xC=\left[\begin{matrix}0&{0}\\0&\langle{\bf v},x{\bf v}\rangle\end{matrix}\right].
\]
Thus the condition \eqref{eq:irreducibility_sufficiency} is satisfied. By Proposition \ref{prop:irreducibility_1-dim}, if the OQRW has faithful initial state $\rho^{(0)}$, the associated QMC is irreducible.
\end{ex}

\subsection{Classical Markov chains}\label{subsec:CMC}
In this subsection we consider the classical Markov chains. The recovery of the classical Markov chains from the OQRWs was introduced in \cite{APSS}.
Let $\mathcal H=\mathbb C$ and $\mathcal K=l^2(\Lambda)$. Then $\mathcal H\otimes \mathcal K\approx l^2(\Lambda)$. Let $P=(P(i,j))_{i,j\in \Lambda}$ be a stochastic matrix, i.e., all the components are nonnegative and satisfy
\[
\sum_{j\in \Lambda}P(i,j)=1 \text{ for all }i\in \Lambda.
\]
For each $i,j\in \Lambda$, let $U_j^i$ be a unitary operator on $\mathcal H=\mathbb C$. (Thus $U_j^i$ is a complex number with modulus $1$, and in the sequel, it turns out that there is no difference with the choice $U_j^i\equiv 1$.) Define
\[
B_j^i:=\sqrt{P(j,i)}U_j^i,\quad i,j\in \Lambda.
\]
We see that
\[
\sum_i{B_j^i}^*{B_j^i}=I,\quad j\in \Lambda.
\]
We notice that since $\mathcal H=\mathbb C$ is a one-dimensional space, the algebra $\mathcal B$ consisting of the operators $x=\sum_j x_j\otimes |j\rangle\langle j|$, with  $(x_j)$ a bounded sequence in $\mathbb C$, is a commutative algebra. If $\rho=(\rho_i)_{i\in \Lambda}$ is a state, i.e., a probability measure on $\Lambda$, we denote by $P_{\rho}$ the projection onto the support of $\rho$. Here the support of $\rho$ is the set of $i\in \Lambda$ at which $\rho_i>0$. By a direct computation from \eqref{4.3} we get
\begin{equation}\label{eq:commutative_transition_expectation}
\mathcal E^{(n)}(x\otimes y)=P_{\rho^{(n)}}xPy,
\end{equation}
where
\[
(Py)_j=\sum_i P(j,i)y_i.
\] 
Notice that,  in the classical Markov chain, if $\rho^{(0)}$ is the initial state (a probability measure) then $\rho^{(1)}=\rho^{(0)}P$, i.e.,
\[
\rho^{(1)}_i=\sum_{j\in \Lambda}\rho^{(0)}_jP(j,i),
\]
and
\[
\rho^{(n)}=\rho^{(0)}P^n.
\]
\begin{prop}\label{prop:state_evolution}
For any initial state $\rho^{(0)}$, the $n$th evolution of the open quantum random walk, $\mathcal M^n(\rho^{(0)})$, is $\rho^{(0)}P^n$. Therefore, the evolutions by classical Markov chain and by open quantum random walk are the same.
\end{prop}
\begin{proof}
By induction, it is enough to see $\mathcal M(\rho^{(0)})$.
\begin{eqnarray*}
(\mathcal M(\rho^{(0)}))_i&=&\sum_{j}B_j^i\rho^{(0)}_j{B_j^i}^*\\
&=&\sum_j\rho^{(0)}_jP(j,i)\\
&=&(\rho^{(0)}P)_{i}.
\end{eqnarray*}
The proof is complete.
\end{proof}
Applying the formula \eqref{eq:commutative_transition_expectation} repeatedly we get
\begin{eqnarray}\label{eq:state}
\rho(I\otimes  \cdots \otimes I\otimes \overset{n\mathrm{th}}{x}\otimes I\otimes \cdots)&=&\rho^{(0)}(E_{0]}(I\otimes   \cdots \otimes I\otimes x\otimes I\otimes \cdots))\nonumber\\
&=&  \rho^{(0)}(\mathcal E^{(0)}(I\otimes\mathcal E^{(1)}(I\otimes\cdots \mathcal E^{(n)}(x\otimes I))))\nonumber\\
&=&\rho^{(0)}(P_{\rho^{0}}PP_{\rho^{(1)}}P\cdots P_{\rho^{(n-1)}}PP_{\rho^{(n)}}x)\nonumber\\
&=&\mathbb E_{\rho^{(0)}}[P^nx]=\mathbb E_{\rho^{(n)}}[x].
\end{eqnarray}
The transition expectation thus recovers the classical Markov chain, which was observed in \cite{AK1, AK2}.

Recall that in the classical Markov chain with transition matrix $P$, we say that a state $j$ is accessible from $i$, written $i\to j$, if $P^n(i,j)>0$ for some $n\in \mathbb N$. We  say that $i$ communicates with $j$, written $i\leftrightarrow j$, if $i\to j$ and $j\to i$. The relation $"\leftrightarrow"$ is an equivalence relation. In the case when every states communicate with every other states, we say that the chain is irreducible. Otherwise, it is called reducible \cite{S}. We want to see the reducibility or irreducibility of classical Markov chains also from the view point of quantum Markov chains.  We emphasize here that, by definition, when we discuss the reducibility or irreducibility of QMCs, not only the transition expectations but also the initial states are concerned. 
\begin{prop}\label{prop:reducible}
Suppose a classical Markov chain with transition matrix $P$ is reducible. Then the QMC $(\rho^{(0)},(\mathcal E^{(n)})_{n\ge 0})$ with a suitably chosen initial state (measure) $\rho^{(0)}$ and transition expectations $\mathcal E^{(n)}$ given by \eqref{eq:commutative_transition_expectation} is reducible.
\end{prop}
\begin{proof}
The state space $\Lambda$ of the Markov chain is decomposed as $\Lambda=T\cup(\cup_{k}R_k)$, where $T$ is the set of transient states and $R_k$'s are closed, recurrent communicating classes. If there is a closed, recurrent communicating class, say $R_1$, by the hypothesis of the proposition, it holds that $R_1\neq \Lambda$.  
Let $p:=P_{R_1}$ be the projection onto the set $R_1$, i.e., $P_{R_1}$ is the indicator function $1_{R_1}$ looked as a multiplication operator on $l^2(\Lambda)$, and we consider $p_{[0}=p\otimes p\otimes\cdots$.  Let $\rho^{(0)}$ be a state (measure) supported on $R_1$. Since $R_1$ is a closed communicating class, $\rho^{(n)}$ is also supported on $R_1$ for all $n\ge 1$. Now the condition $\rho^{(n)}_jp(j)=\rho^{(n)}_j$ is equivalent to saying that $p(j)=1$ on the support of $\rho^{(n)}$, and this is the case by our construction. Therefore by Theorem \ref{thm4.8} the QMC is reducible. If there is no closed, recurrent communicating class, then the set $\Lambda$ consists only of transient states. Fix an $i_0\in \Lambda$ and let $C_0$ be the communicating class containing $i_0$. By the assumption $C_0$ is not closed, i.e., there is a state $j\in \Lambda\setminus C_0$ such that $i_1\to j$ for some $i_1\in C_0$ and $j\nrightarrow i$ for all $i\in C_0$. Let $C_1:=\{j'\in \Lambda:j\rightarrow j'\}$. Then $C_1\cap C_0=\emptyset$ and if the initial measure $\rho^{(0)}$ is supported on the set $C_1$, it follows that $\rho^{(n)}$ is also supported on the set $C_1$ for all $n\ge 1$. Defining now $p:=P_{C_1}$, the projection onto the set $C_1$, we see as above that $p_{[0}=p\otimes p\otimes\cdots$ is a reducing projection for the QMC $(\rho^{(0)},(\mathcal E^{(n)})_{n\ge 0})$.
\end{proof}
Let us now consider the converse problem.
\begin{prop}\label{prop:irreducible}
Suppose that the classical Markov chain with transition matrix $P$ is irreducible. Then the QMC of transition expectation \eqref{eq:commutative_transition_expectation} with any faithful initial state is irreducible.
\end{prop}
\begin{proof}
Suppose that the transition expectation \eqref{eq:commutative_transition_expectation} is constructed from a faithful initial state $\rho^{(0)}$. From the assumption of irreducibility of the classical Markov chain, the distribution at any time has full support. This implies by Proposition \ref{prop:state_evolution} that the state $\rho^{(n)}$ is faithful for any $n\ge 0$.  The result now follows from Theorem \ref{thm:irreducibility}.
\end{proof}

\appendix
\section{Equivalence of concepts of reducibility/irreducibility of OQRWs defined in \cite{CP} and in this paper}\label{sec:comparison}

First of all we recall the definition of reducibility/irreducibility used in \cite{CP}. Let $\Phi$ be a positive map on the ideal $\mathcal I_1(\mathfrak h)$ of trace class operators on a Hilbert space $\mathfrak h$. When we come to our model, $\mathfrak h$ is $\mathcal H\otimes \mathcal K$ and $\Phi$ is $\mathcal M$. $\Phi$ is said to be irreducible (see \cite[Definition 3.1]{CP}) if the only orthogonal projections $p$ reducing $\Phi$, i.e. such that $\Phi(p\mathcal I_1(\mathfrak h)p)\subset p\mathcal I_1(\mathfrak h)p$, are $p=0$ and $I$. Applying to OQRWs, Carbone and Pautrat have shown (terminology in our language):
\begin{prop} (\cite[Proposition 3.8]{CP})\label{prop:CP_condition}
The completely positive and trace preserving map $\mathcal M$ is irreducible if and only if for any $i,j\in \Lambda$ and any $\psi,\xi\in \mathcal H\setminus\{0\}$, there is a path $\pi\in \mathcal P(i,j)$ such that $\langle \xi,B_\pi\psi\rangle\neq 0$.
\end{prop}
Now we show the definitions of reducibility/irreducibility of OQRWs given in \cite{CP} and in the present paper are equivalent. First we remark that as given by \cite[Proposition 6.1, item 3]{CP}, once an OQRW is reducible (in the sense of \cite{CP}) one can always find a reducing projection $p$ of the block-diagonal form: $p=\sum_j p(j)\otimes |j\rangle\langle j|$. Conversely speaking, if there is no nontrivial block-diagonal reducing projection the OQRW is irreducible. Suppose the OQRW is reducible in the sense of \cite{CP} with a reducing projection $p=\sum_j p(j)\otimes |j\rangle\langle j|$. By \cite[Proposition 6.2]{CP}, it holds that for any $i,j\in \Lambda$, 
\begin{equation}\label{eq:CP_condition}
B_j^ip(j)=p(i)B_j^ip(j).
\end{equation}
Take an initial state $\rho^{(0)}=\sum_j\rho^{(0)}_j\otimes |j\rangle \langle j|$ such that $p(j)\rho^{(0)}_jp(j)=\rho^{(0)}_j$ for all $j\in \Lambda$. We can show by induction that for all $n\ge 0$ and $j\in \Lambda$,
\begin{equation}\label{eq:DKY_condition} 
p(j)\rho^{(n)}_jp(j)=\rho^{(n)}_j. 
\end{equation}
In fact, suppose \eqref{eq:DKY_condition} holds for $n=0,\cdots,k$. Then, by the assumption hypothesis and \eqref{eq:CP_condition}
\begin{eqnarray*}
p(j)\rho^{(k+1)}_jp(j)&=&\sum_ip(j)B_i^j\rho^{(k)}_i{B_i^j}^*p(j)\\
&=&\sum_ip(j)B_i^jp(i)\rho^{(k)}_ip(i){B_i^j}^*p(j)\\
&=&\sum_iB_i^jp(i)\rho^{(k)}_ip(i){B_i^j}^*\\
&=&\sum_iB_i^j\rho^{(k)}_i{B_i^j}^*=\rho^{(k+1)}_j.
\end{eqnarray*}
Now \eqref{eq:DKY_condition} holds and by Theorem \ref{thm4.8} the OQRW is reducible in the sense of this paper (recall \eqref{eq:DKY_condition} is equivalent to $\rho^{(n)}_jp(j)=\rho^{(n)}_j$).

Conversely, suppose that the OQRW is reducible in the sense of present paper. By Theorem \ref{thm4.8}, there is a nontrivial projection $p=\sum_j p(j)\otimes |j\rangle\langle j|$ such that \eqref{eq:DKY_condition} holds for $n\ge n_0$ for some $n_0$. Find a $j\in \Lambda$ such that $p(j)\neq I_\mathcal H$. By the assumption we have for any $k\ge 0$,
\[
\mathrm{Tr}(\rho^{(n_0+k)}_jp(j)^\perp)=\mathrm{Tr}(\rho^{(n_0+k)}_jp(j)p(j)^\perp)=0.
\]
Take an $i\in \Lambda$ such that $\rho^{(n_0)}_i\neq 0$. From the above relation we have 
\begin{eqnarray*}
0&=&\mathrm{Tr}(\rho^{(n_0+k)}_jp(j)^\perp)\\
&=&\sum_{i_0,\cdots,i_{k-1}}\mathrm{Tr}\left(B_{\pi(i_0,\cdots,i_{k-1},j)}\rho^{(n_0)}_{i_0}B^*_{\pi(i_0,\cdots,i_{k-1},j)}p(j)^\perp\right)\\
&\ge &\mathrm{Tr}\left(B_{\pi }\rho^{(n_0)}_{i}B^*_{\pi }p(j)^\perp\right)=\mathrm{Tr}\left(\rho^{(n_0)}_{i}B^*_{\pi }p(j)^\perp B_{\pi }\right)\ge 0,
\end{eqnarray*}
for any path $\pi\in \mathcal P(i,j)$ of length $k$. Thus for any  $0\neq \psi\in \mathcal H$ lying in the spectral projection of $\rho^{(n_0)}_i$ away from zero, e.g., any eigenvector of $\rho^{(n_0)}_i$ corresponding to nonzero eigenvalue,
\[
\langle \psi,B^*_{\pi }p(j)^\perp B_{\pi }\psi\rangle=0.
\]
Therefore, for any such a vector $0\neq \psi$     and $0\neq \xi\in p(j)^\perp$, and for any path $\pi\in \mathcal P(i,j)$,
\begin{eqnarray*}
|\langle \xi,B_\pi \psi\rangle|&=&|\langle \xi,p(j)^\perp B_\pi \psi\rangle|\\
&\le &\|\xi\|\langle p(j)^\perp B_\pi \psi,p(j)^\perp B_\pi \psi\rangle^{1/2}=0.
\end{eqnarray*}
By Proposition \ref{prop:CP_condition}, it says that the OQRW is reducible in the sense of \cite{CP}. This completes the proof of equivalence.\\[1ex]

\noindent\textbf{Acknowledgments}. We are grateful to anonymous referees for giving many valuable comments. It improved the paper a lot. A. Dhahri acknowledges support by Basic Science Research Program through the National Research Foundation of Korea (NRF) funded by the Ministry of Education (grant 2016R1C1B1010008). The research by H. J. Yoo was supported by Basic Science Research Program through the National
Research Foundation of Korea (NRF) funded by the Ministry of Education (NRF-2016R1D1A1B03936006).

\end{document}